%% file: parity.tex
\documentclass{article}
\input{stdinc.tex}

\usepackage{mdwlist}
\usepackage{algorithm}
\usepackage{algpseudocode}

\newcommand{\lpn}{$k$-\textsc{LPN}\xspace}
\newcommand{\park}{\textsc{PAR}($k$)\xspace}

\title{On learning $k$-parities with and without noise}
\author{Arnab Bhattacharyya\thanks{Supported in part by DST Ramanujan Fellowship. Email: \href{}{\texttt{arnabb@csa.iisc.ernet.in}}} \qquad 
Ameet Gadekar\thanks{Email: \href{}{\texttt{ameet.gadekar@csa.iisc.ernet.in}}}
\qquad Ninad Rajgopal\thanks{Email: \href{}{\texttt{ninad.rajgopal@csa.iisc.ernet.in}}} \\[1em]
Department of Computer Science \& Automation \\
{Indian Institute of Science}}

\begin{document}
\maketitle

\begin{abstract}
We first consider the problem of learning $k$-parities in the on-line mistake-bound  model: given a hidden vector $x \in \zo^n$ with $|x|=k$ and a sequence of ``questions" $a_1, a_2, \dots \in \zo^n$, where the algorithm must reply to each question with $\langle a_i, x\rangle \pmod 2$, what is the best tradeoff between the number of mistakes made by the algorithm and its time complexity? We improve the previous best result of Buhrman et.~al.~\cite{BGM10} by an $\exp(k)$ factor in the time complexity.

Second, we consider the problem of learning $k$-parities in the presence of  classification noise of rate $\eta \in (0,\nfrac12)$.  A polynomial time algorithm for this problem (when   $\eta > 0$ and $k = \omega(1)$) is a longstanding challenge in  learning theory. Grigorescu et al.~\cite{GRV11} showed an algorithm   running in time ${n \choose k/2}^{1 + 4\eta^2 +o(1)}$. Note that  this algorithm inherently requires time ${n \choose k/2}$ even when  the noise rate $\eta$ is polynomially small. We observe that for sufficiently small noise rate, it is possible to break the $n \choose k/2$ barrier. In particular, if for some function $f(n) = \omega(1)$ and $\alpha \in [\nfrac{1}{2}, 1)$,  $k = n/f(n)$ and $\eta = o(f(n)^{- \alpha}/\log n)$, then there is an algorithm for the problem with running time $ \poly(n)\cdot {n \choose k}^{1-\alpha} \cdot e^{-k/4.01}$.
\end{abstract}

\section{Introduction}\label{sec:intro}
By now, the ``Parity Problem" of Blum, Kalai and Wasserman \cite{BKW03} has acquired widespread notoriety. The question is simple enough to be in our second sentence: in order to learn a hidden vector $x \in \zo^n$, what is the least number of random examples $(a, \ell)$ that need to be seen, where $a$ is uniformly chosen from $\zo^n$ and $\ell = \sum_i a_i x_i \pmod 2$ with probability at least $1-\eta$? Information-theoretically, $x$ can be recovered only after $O(n)$ examples, even if the noise rate $\eta$ is close to $1/2$. But if we add the additional constraint that the running time of the learning algorithm be minimized, the barely subexponential running time of \cite{BKW03}'s algorithm, $2^{O(n/\log n)}$, still holds the record of being the fastest known for this problem!

Learning parities with noise is a central problem in theoretical computer science. It has incarnations in several different areas of computer science, including coding theory as the problem of learning random binary linear codes and cryptography as the ``learning with errors" problem that underlies lattice-based cryptosystems \cite{Regev09, BV11}. In learning theory, the special case of the problem where the hidden vector $x$ is known to be supported on a set of size $k$ much smaller than $n$ has great relevance. We refer to this problem as {\em learning $k$-parity with noise} or \lpn.
Feldman et al. \cite{FGKP09} showed that learning $k$-juntas, as well as learning $2^k$-term DNFs from uniformly random examples and variants of these problems in which the noise is adversarial instead of random, all reduce to the \lpn problem. For the \lpn problem, the current record is that of Grigorescu, Reyzin and Vempala \cite{GRV11} who showed a learning algorithm that succeeds with constant probability, takes ${n \choose k/2}^{1+(2\eta)^2+o(1)}$ time and uses $\frac{k \log n}{(1-2\eta)^2} \cdot \omega(1)$ samples. When the noise rate $\eta$ is close to $1/2$, this running time is improved by an algorithm due to G.~Valiant \cite{Valiant12} that runs in time $n^{0.8k} \cdot \poly(\frac{1}{1-2\eta})$. It is a wide open challenge to find a polynomial time algorithm for \lpn for growing $k$ or to prove a negative result.

Another outstanding challenge in machine learning is the problem of learning parities without noise in the ``attribute-efficient" setting \cite{Blum96}. The algorithm is given access to a source of examples $(a,\ell)$ where $a$ is chosen uniformly from $\zo^n$ and $\ell = \sum_i a_i x_i \pmod 2$ with no noise, and the question is to learn $x$ while simultaneously reducing the time complexity and the number of examples drawn by the algorithm. Again, we focus on the case where $x$ has sparsity $k \ll n$. Information-theoretically, of course, $O(k \log n)$ examples should be sufficient, as each linearly independent example reduces the number of consistent $k$-parities by a factor of $2$. But the fastest known algorithm making $O(k \log n)$ samples runs in time $\tilde{O}({n \choose k/2})$ \cite{KS06}, and it is open whether there exists a polynomial time algorithm for learning parities that is attribute-efficient, i.e. it makes $\poly(k \log n)$ samples. Buhrman, Garc{\'i}a-Soriano and Matsliah \cite{BGM10} give the current best tradeoffs between the sample complexity and running time for learning parities in this noiseless setting. Notice that with $O(n)$ samples, it is easy to learn the $k$-parity in polynomial time using Gaussian elimination.

\subsection{Our Results}

We first study the noiseless setting. Our main technical result is an improved tradeoff between the sample complexity and runtime for learning parities. 

\begin{theorem}\label{thm:pac}
Let $k, t: \N \to \N$ be two functions\footnote{We assume throughout, as in \cite{BGM10}, that $k$ and $t$ are constructible in quadratic time.} satisfying $\log \log n \ll k(n) \ll t(n) \ll n$. For any $\delta > 0$, there is an algorithm that learns the concept class of $k$-parities on $n$ variables with confidence parameter $\delta$, using $O(kn/t + \log {t \choose k} + \log(1/\delta))$ uniformly random examples and $e^{-k/4.01} {t \choose k} \cdot \poly(n)\cdot \log(1/\delta)$ running time\footnote{The ``4.01" can be replaced by any constant more than $4$. }.
\end{theorem}

Actually, we prove our result in the mistake-bound model \cite{Littlestone89} that is stronger than the PAC model discussed above (in fact, strictly stronger assuming the existence of one-way functions \cite{Blum94}). As a consequence, a theorem of the above form also holds when the examples come from an arbitrary distribution. We defer the statement of the result for the mistake-bound model to \cref{sec:nless}.

For comparison, let us quote the closely related result of Buhrman et. al.:

\begin{theorem}[Theorem 2.1 of \cite{BGM10}]\label{thm:bgm}
Let $k, t: \N \to \N$ be two functions satisfying $k(n) \leq t(n) \leq n$. For any $\delta > 0$, there is an algorithm that learns the concept class of $k$-parities on $n$ variables with confidence parameter $\delta$, using $O(kn/t + \log {t \choose k} + \log(1/\delta))$ uniformly random examples and ${t \choose k} \cdot \poly(n)\cdot \log(1/\delta)$ running time.
\end{theorem}
Thus, in the comparable regime, our \cref{thm:pac} improves the runtime complexity of \cref{thm:bgm} by an $\exp(k)$ factor while its sample complexity remains the same upto constant factors. Note that as $t$ approaches $k$, our algorithm makes $O(n)$ samples and takes $\poly(n)$ time which is the complexity of the Gaussian elimination approach. On the other hand, if $t=n/\log(n/k)$, our algorithm makes $O(k \log(n/k))$ samples and takes\footnote{By $\exp(\cdot)$, we mean $2^{O(\cdot)}$.} $\exp(-k)\cdot {n/k \choose k}$ time (ignoring polynomial factors), compared to the trivial approach which explicitly keeps track of the subset of all the $k$-weight parities consistent with examples given so far and which makes $O(k \log(n/k))$ samples and takes $O({n \choose k})$ time.

We next examine the noisy setting. Here, our contribution is a simple, general observation that does not seem to have been explicitly made before.
\begin{theorem}\label{thm:red}
Given an algorithm $\cA$ that learns \park over the uniform distribution with confidence parameter $\delta$ using $s(\delta)$ samples and running time $t(\delta)$,there is an algorithm $\cA'$ that solves the \lpn problem with noise rate $\eta \in (0,\nfrac13)$, using $O(s(\delta/2)\log(1/\delta))$ examples and running time $\exp(O(H(3\eta/2)\cdot s(\delta/2) \cdot \log(1/\delta)))) \cdot (t(\delta/2) + s(\delta/2)\log(1/\delta))$ and with confidence parameter $\delta$. 
\end{theorem}
In the above, $H: [0,1] \to [0,1]$ denotes the binary entropy function $H(p)  = p\log_2\frac1p + (1-p) \log_2\frac1{1-p}$. The main conceptual message carried by \cref{thm:red} is that improving the sample complexity for efficient learning of noiseless parity improves the running time for learning of noisy parity. For instance, if we use Spielman's algorithm as $\cA$, reported in \cite{KS06}, that learns $k$-parity using $O(k \log n)$ samples and $O({n \choose k/2})$ running time, we immediately get the following:
\begin{corollary}\label{cor:att}
For any $\eta \in (0,\nfrac13)$ and constant confidence parameter, there is an algorithm for \lpn with sample complexity $O(k \log n)$ and running time ${n \choose k/2}^{1 + O(H(1.5\eta))}$. 
\end{corollary}
For comparison, the current best result of \cite{GRV11} has runtime ${n \choose k/2}^{1+4\eta^2 + o(1)}$ and sample complexity $\omega(k \log n)$. In the regime under consideration, our algorithm's runtime has a worse exponent but an asymptotically better sample complexity. 

The result of \cite{GRV11} requires ${n \choose k/2}$ time regardless of how small $\eta$ is. We show via \cref{thm:red} and \cref{thm:pac} that it is possible to break the ${n \choose k/2}$ barrier when $\eta$ is a small enough function of $n$. 
\begin{corollary} \label{cor:mainapp}
Suppose $k(n) = n/f(n)$ for some function $f: \N \to \N$ for which $f(n) \ll n/\log\log n$, and suppose $\eta(n) = o(\frac{1}{((f(n))^\alpha \log n)})$ for some $\alpha\in [1/2,1)$. Then, for constant confidence parameter, there exists an algorithm for \lpn with noise rate $\eta$ with running time $e^{-k/4.01 + o(k)} \cdot {n \choose k}^{1-\alpha} \cdot \poly(n)$ and sample complexity $O(k (f(n))^\alpha)$.
\end{corollary}

We note that because of the results of Feldman et.~al. \cite{FGKP09}, the above results for \lpn also extend to the setting  where the example source adversarially mislabels examples instead of randomly but with the same rate $\eta$.

\subsection{Our Techniques}
We first give an algorithm to learn parities in the noiseless setting in the mistake bound model. We use the same approach as that of \cite{BGM10} (which was itself inspired by \cite{APY09}). The idea is to consider a family $\cS$ of subsets of $\zo^n$ such that the hidden $k$-sparse vector is contained inside one of the elements of $\cS$. We maintain this invariant throughout the algorithm. Now, each time an example comes, it specifies a halfspace $H$ of $\zo^n$ inside which the hidden vector is lying. So, we can update $\cS$ by taking the intersection of each of its elements with $H$. If we can ensure that the set of points covered by the elements of $\cS$ is decreasing by a constant factor at every round, then after $O(\log \sum_{S \in \cS} |S|)$ examples, the hidden vector is learned. The runtime is determined by the number of sets in $\cS$ times the cost of taking the intersection of each set with a halfspace.

One can think of the argument of Buhrman et al. \cite{BGM10} as essentially initializing $\cS$ to be the set of all ${n \choose k}$ subspaces spanned by $k$ standard basis vectors. The intersections of these subspaces with a halfspace can be computed efficiently by Gaussian elimination. Our idea is to reduce the number of sets in $\cS$. Note that we can afford to make the size of each set in $\cS$ larger by some factor $C$ because this only increases the sample complexity by an additive $\log C$. Our approach is (essentially) to take $\cS$ to be a random collection of subspaces spanned by $\alpha k$ standard basis vectors, where $\alpha > 1$ is a sufficiently large constant. We show that it is sufficient for the size of $\cS$ to be smaller than ${n/\alpha \choose k}$ by a factor that is exponential in $k$, so that the running time is also improved by the same factor. Moreover, the sample complexity increases by only a lower-order additive term.

Our second main contribution is a reduction from noiseless parity learning to noisy parity learning. The algorithm is a simple exhaustive search which guesses the location of the mislabelings, corrects those labels, applies the learner for noiseless parity and then verifies whether the output hypothesis matches the examples by drawing a few more samples. Surprisingly, this seemingly immediate algorithm allows us to devise the first algorithm which has a better running time than ${n \choose k/2}$ in the presence of a non-trivial amount of noise. The lesson seems to be that if we hope to beat ${n \choose k/2}$ for constant noise rates, we should first address the open question of Blum \cite{Blum96} of devising an attribute-efficient algorithm to learn parity without noise.

\section{Preliminaries}

Let \park be the class of all $f \in \zo^n$ of Hamming weight $k$. So, $|\text{\park}| = {n \choose k}$. With each vector $f \in \text{\park}$, we associate a parity function $f: \zo^n \to \zo$ defined by $f(a) = \sum_{i=1}^n x_i a_i \pmod 2$. 

Let $\cC$ be a concept class of Boolean functions on $n$ variables, such as \park. We discuss two models of learning in this work. One is Littlestone's {\em online mistake bound} model \cite{Littlestone89}. Here, learning proceeds in a series of rounds, where in each round, the learner is given an unlabeled boolean example $a \in \zo^n$ and must predict the value $f(a)$ of an unknown target function $f \in \cC$. Once the learner predicts the value of $f(a)$, the true value of $f(a)$ is revealed to the learner by the teacher. The {\em mistake bound} of a learning algorithm is the worst-case number of mistakes that the algorithm makes over all sequences of examples and all possible target functions $f \in \cC$. 

The second model of learning we consider is Valiant's famous {\em PAC model} \cite{Valiant84}  of learning from random examples. Here, for an unknown target function $f \in \cC$, the learner has access to a source of examples $(a, f(a))$ where $a$ is chosen independently from a distribution $\cD$ on $\zo^n$. A learning algorithm is said to PAC-learn $\cC$ with {\em sample complexity} $s$, {\em running time} $t$, {\em approximation parameter} $\eps$ and {\em confidence parameter} $\delta$ if for all distributions $\cD$ and all target functions $f \in \cC$, the algorithm draws at most $s$ samples from the example source, runs for time at most $t$ and outputs a function $f^*$ such that, with a probability at least $1-\delta$:
$$\Pr_{a \leftarrow \cD}[f(a) \neq f^*(a)]<\eps$$

Often in this paper (e.g., all of the Introduction), we consider PAC-learning over the uniform distribution, in which case $\cD$ is fixed to be uniform on $\zo^n$. Notice that for learning \park over the uniform distribution, we can take $\eps = \frac12$ because any two distinct parities differ on half of $\zo^n$.

There are standard conversion techniques  which can be used to transform any mistake-bound algorithm into a PAC learning algorithm (over arbitrary distributions):
\begin{theorem}[\cite{Angluin88, Haussler88, Littlestone89}] Any algorithm $\cA$ that learns a concept class $\cC$ in the mistake-bound model with mistake bound $m$ and running time $t$ per round can be converted into an algorithm $\cA'$ that PAC-learns $\cC$ with sample complexity $O(\frac1\eps m + \frac1\eps \log \frac1\delta)$, running time $O(\frac1\eps mt+ \frac t\eps \log\frac1\delta)$, approximation parameter $\eps$, and confidence parameter $\delta$.
\end{theorem}

The \lpn problem with noise rate $\eta$, introduced in \cref{sec:intro}, corresponds to the problem of PAC-learning \park under the uniform distribution, when the example source can mislabel examples with a rate $\eta \in (0,\nfrac12)$.  More generally, one can study the {\em \lpn problem over $\cD$}, an arbitrary distribution. \cite{GRV11} show the following for this problem:
\begin{theorem}[Theorem 5 of \cite{GRV11}]
For any $\eps, \delta, \eta \in (0,\nfrac12)$, and distribution $\cD$ over $\zo^n$, the \lpn problem over $\cD$ with noise rate $\eta$  can be solved using $\frac{k\log(n/\delta)\omega(1)}{\eps^2(1-2\eta)^2}$ samples in time $ \frac{1}{\eps^2 (1-2\eta)^2} \cdot {n \choose k/2}^{1+(\frac \eta {\eps+\eta -2\eps\eta})^2  + o(1)}$, where $\eps$ and $\delta$ are the approximation and confidence parameters respectively. 
\end{theorem}

\section{In the absence of noise}\label{sec:nless}
We state the main result of this section.
\begin{theorem} \label{thm:imBGM}
Let $k,t : \mathbb{N} \rightarrow \mathbb{N}$ be two functions such that $\log\log n \ll k(n) \ll t(n) \ll n$. Then for every $n \in \mathbb{N}$, there is an algorithm that learns \park in the mistake-bound model, with mistake bound at most $(1+o(1))\frac{kn}{t} +  \log {t \choose k}$ and running time per round $ e^{-k/4.01} \cdot {t \choose k} \cdot \tilde{O}\left(\left({kn}/{t} \right)^2\right)$.
\end{theorem}
Using \cref{thm:red}, we directly obtain \cref{thm:pac}. In fact, since \cref{thm:red} produces a PAC-learner over any distribution, a statement of the form of \cref{thm:pac} holds for examples obtained from any distribution.

For comparison, we quote the relevant result of \cite{BGM10} in the mistake-bound model.
\begin{theorem}[Theorem 2.1 of \cite{BGM10}]\label{thm:bgm}
Let $k,t : \mathbb{N} \rightarrow \mathbb{N}$ be two functions such that $ k(n) \leq t(n) \leq n$. Then for every $n \in \mathbb{N}$, there is a deterministic algorithm that learns \park in the mistake-bound model, with mistake bound at most $k\lceil\frac{n}{t}\rceil +  \log {t \choose k}$ and running time per round ${t \choose k} \cdot O(\left({kn}/{t} \right)^2)$.
\end{theorem}
Note that their mistake bound is better by a lower-order term which we do not see how to avoid in our setup. This slack is not enough though to recover \cref{thm:imBGM} from \cref{thm:bgm}: dividing $t$ by $C$ roughly multiplies the sample complexity by $C$ and divides the running time by $C^k$ in \cite{BGM10}'s algorithm, whereas in our algorithm, dividing $t$ by $C$ roughly multiplies the sample complexity by $C$ and divides the running time by $(1.28C)^k$.

\ignore{
Basically, the above result says that for any $\epsilon >0$, we get an exponential in $k$ improvement in running time over that of the BGM algorithm. Let us look at few interesting values of $\epsilon$ for the theorem.
\begin{itemize}
\item[$\bullet$] When $\epsilon = 0$, then we do not get any improvement. The running time per round and the mistake bound that we get for this case are same as those of the BGM algorithm.
\item[$\bullet$] When $\epsilon = 1$, then we obtain an improvement of $O(2^{-k/3})$ in the running time per round over that of the BGM algorithm. Note that, the mistake bound is now at most twice that of the BGM algorithm.
\item[$\bullet$] For $0 < \epsilon <1$, there is still an exponential in $k$ improvement in the running time which depends on $\epsilon$ since as $\epsilon \rightarrow 0$, the improvement term $\Big(\frac{1+ \epsilon}{1+ 3 \epsilon} \Big) \rightarrow 1$.
\end{itemize}

From the results of ANG88, LIT89, our main result imply an improvement in the running time of $PAR(k)$ in the PAC model. Further, the number of samples required for this algorithm in PAC model are not much large compared to that of BGM algorithm in PAC model.

In the next section, we prove the \cref{thm:imBGM} by describing an algorithm.
}
\subsection{The Algorithm}
Let $f \in \zo^n$ be the hidden vector of sparsity $k$ that the learning algorithm is trying to learn. Let $e=\lbrace e_1, e_2, \cdots, e_n \rbrace$ be the set of standard basis of the vector space $\lbrace 0,1 \rbrace^n$. 

Let $\alpha$ be a large constant we set later, and let $T = \alpha t$.  Note that $T \ll n$. We define an arbitrary partition $\pi = C_1, C_2, \cdots, C_T$ on the set $e$ into $T$ parts, each of size at most $\lceil n/T\rceil$. Next, let $S_1, \dots, S_m \subset [T]$ be $m$ random subsets of $[T]$, each of size $\alpha k$. We choose $m$ to ensure the following:
\begin{claim}\label{clm:simp}
If $m = \tilde{O}\left(\frac{{T \choose \alpha k}}{{T-k \choose \alpha k - k}}\right)$, then with nonzero probability, for every set $A \subset [T]$ of size $k$, $A \subset S_i$ for some $i \in [m]$. 
\end{claim}
\begin{proof}
This follows from the simple observation that for any fixed $i \in [m]$, $\Pr[A \subset S_i] = {T-k \choose \alpha k - k}\big/{T \choose \alpha k}$, and so, 
$$\Pr[\exists i \in [m], A \not \subset S_i] = \left(1 - {T-k \choose \alpha k - k}\bigg/{T \choose \alpha k}\right)^m
 \leq e^{- m {T-k \choose \alpha k - k}/{T \choose \alpha k}}$$
Choosing $m = 2\frac{{T \choose \alpha k}}{{T-k \choose \alpha k - k}} \log {T \choose k}$ and applying the union bound finishes the proof.
\end{proof}
 
We fix some choice of $S_1, \dots, S_m \subset [T]$ that satisfies the conclusion of Claim \ref{clm:simp} for what follows. In fact, the rest is exactly \cite{BGM10}'s algorithm, which we reproduce for completeness. 

For every $i \in [m]$, let $M_i \subset \zo^n$ be the span of $\bigcup_{j \in S_i} C_j$. Note that $\left|\bigcup_{j \in S_i} C_j\right| \leq \alpha k \lceil n/T \rceil \leq \alpha k \cdot \left(\frac n T + 1 \right)\leq \frac{kn}{t} + \alpha k = (1+o(1)) kn/t$, as $t \ll n$ and $\alpha$ is a constant. So, $M_i$ is a linear subspace containing at most $2^{(1+o(1))kn/t}$ points.

Note that every $f \in \zo^n$ with $|f|=k$ is contained in some $M_i$. This is simply because every set of $k$ standard basis vectors is contained in at most $k$ of the $T$ parts in the partition $\pi$, and by Claim \ref{clm:simp}, every subset of $[T]$ of size $k$ is contained in some $S_i$. 

Initially, the unknown target vector $f$ can be in any of the $M_i$'s. Consider what happens when the learner sees an example $a \in \zo^n$ and a label $y \in \zo$. For $i \in [m]$, let $M_i(a,y) = \{f \in M_i : f(a) = y\}$.  $M_i(a,y)$ may be of size $0$, $|M_i|$ or $|M_i|/2$. Note that the size of $M_i(a,y)$ can be efficiently found using Gaussian elimination.

We are now ready to describe the algorithm:
\begin{itemize}
\item \textbf{Initialization:}
The learning algorithm begins with a set of affine spaces $N_i, i \in [m]$ represented by a system of linear equations. Initialize the affine spaces $N_i = M_i$ for all $i \in [m]$.

\item \textbf{On receiving an example $a \in \zo^n$:}
Predict its label $\hat{y} \in \zo$ such that $\sum_{i \in [m]} |N_i(a,\hat{y})| \geq \sum_{i \in [m]} |N_i(a,1-\hat{y})|$. 

\item
\textbf{On receiving the answer from the teacher $y = f(a)$:}
Update $N_i$ to $N_i(a,y)$ for each $i \in [m]$. 
\end{itemize}

\subsection{Analysis}
Before we analyze the algorithm, we first establish a combinatorial claim that is the crux of our improvement:
\begin{lemma}\label{lem:calc}
If $\alpha$ is a large enough constant,
$$\frac{{T \choose \alpha k}}{{T-k \choose \alpha k - k}} \leq e^{-k/4.01} \cdot {t \choose k}$$
\end{lemma}
\begin{proof}
\begin{align*}
\displaybreak[1]
\frac{1}{{t \choose k}} \cdot \frac{{T \choose \alpha k}}{{T-k \choose \alpha k - k}}
&= \prod_{i=0}^{k-1} \frac{k-i}{t-i} \cdot \frac{T-i}{\alpha k - i}\\
&= \prod_{i=0}^{k-1} \frac{\alpha t-i}{\alpha k - i} \cdot \frac{k-i}{t-i} \\
& = \prod_{i=1}^{k-1} \left(1 - \frac{i\left(1-\frac1\alpha\right)\left(\frac{1}{k-i}-\frac{1}{t-i}\right)} {1 + \frac{i}{k-i}\left(1-\frac1\alpha\right)}\right)\\
&\leq \prod_{i=1}^{k-1} \left(1 - \frac{0.999}{1+\frac{k-i}{i}\frac{\alpha}{\alpha-1}}\right)
\end{align*}
where the equalities are routine calculation and the inequality is using that $k(n) \ll t(n)$. Each individual term in the product is strictly less than $1$. So, the above is bounded by:
\begin{align*}
&\leq \prod_{i=k/(2-\eps)}^{k-1} \left(1 - \frac{0.999}{1+\frac{k-i}{i}\frac{\alpha}{\alpha-1}}\right)\\
&\leq \left(1-\frac{0.999}{1+(1-\eps)\frac{\alpha}{\alpha-1}}\right)^{\frac{1-\eps}{2-\eps}k}\\
&\leq \exp\left(- \lg e \cdot \frac{0.999(1-\eps)}{(2-\eps)(1+(1-\eps)\frac{\alpha}{\alpha-1})}k\right) \leq e^{-k/4.01}
\end{align*}
for a small enough constant $\eps > 0$ and large enough constant $\alpha > 1$.
\end{proof}

\begin{proof}[Proof of \cref{thm:imBGM}]
Fix $\alpha$ to be a constant that makes the conclusion of \cref{lem:calc} true.

We first check that the invariant is maintained throughout the algorithm that $f \in \cup_{i \in [m]} N_i$. This holds at initiation by the argument given earlier. After that, obviously, if $f \in N_i$, then $f \in N_i(a, f(a))$ for any $a \in \zo^n$, and so the invariant holds. Therefore, if the algorithm terminates, it will find the hidden vector $f$ and return it as the solution. The rate of convergence is precisely captured by the number of mistakes learning algorithm makes, which we describe next.

\paragraph{Mistake Bound.}
Notice that when the algorithm begins, the sum of the sizes of all the affine spaces, $\sum_{i} |N_i| \leq \tilde{O} \bigg(\frac{{T \choose \alpha k}}{{T-k \choose \alpha k -k}}\bigg) 2^{(1+o(1)) k n/t }$. Now whenever the learner makes a mistake by predicting $\hat{y} \neq y$, the size of all affine spaces $\sum_{i} |N_i|$ reduces by a factor of at least $2$. This is due to the definition of $\hat{y}$ and the fact that $|N_i(a,\hat{y})| + |N_i(a,1- \hat{y})| = |N_i|$. 

Hence, using \cref{lem:calc}, after at most 
$$\log \left(\sum_{i} |N_i| \right) \leq \log \left[\tilde{O} \bigg(\frac{{T \choose \alpha k}}{{T-k \choose \alpha k -k}}\bigg) 2^{(1+o(1)) k n/t }\right] \leq (1+o(1))kn/t + \log {t \choose k} - \Omega(k) + \log O\left(\log {t \choose k}\right)$$ mistakes, the size of $\cup_{s \in S} N_s$ will decrease to 1, which by the invariant above will imply that $\cup_{s \in S} N_s = \lbrace f \rbrace$, and hence the learner makes no more mistakes. . Since we assume $k \gg \log \log n$ and $t \ll n$, we can bound the number of mistakes by:
$(1+o(1))kn/t + \log {t \choose k}$

\paragraph{Running Time.}
We analyze the running time of the learner for each round. At each round, for a question $a \in \zo^n$, we need to compute $|N_i(a,0)|$ and $|N_i(a,1)|$ as well as store a representation of the updated $N_i$. Now, since for each $N_i$ is spanned by at most $\ell = (1+o(1))kn/t$ basis vectors, we can treat each $N_i$ as a linear subspace in $\zo^\ell$. $N_i(a,0)$ and $N_i(a,1)$ can be computed by performing Gaussian elimination on a system of linear equations involving $\ell$ variables, which takes $O(\ell^2)$ time. Thus, the total running time is $O(m \ell^2)$, which using \cref{lem:calc} is exactly the bound claimed in \cref{thm:imBGM}.
\end{proof}

\section{In the presence of noise}
Recall the \lpn problem. In this section, we show a reduction from \lpn to noiseless learning of \park  and its applications.

\subsection{The Reduction}
We focus on the case when the noise rate $\eta$ is bounded by a constant less than half.\\

{\noindent
\textbf{\cref{thm:red} (recalled)}
{\em Given an algorithm $\cA$ that learns \park over the uniform distribution with confidence parameter $\delta$ using $s(\delta)$ samples and running time $t(\delta)$,there is an algorithm $\cA'$ that solves the \lpn problem with noise rate $\eta \in (0,\nfrac13)$, using $O(s(\delta/2)\log(1/\delta))$ examples and running time $\exp(O(H(3\eta/2)\cdot s(\delta/2) \cdot \log(1/\delta)))) \cdot (t(\delta/2) + s(\delta/2)\log(1/\delta))$ and with confidence parameter $\delta$.  }}\\

Let $\cA(\delta)$ be a PAC-learning algorithm over the uniform distribution for \park of length $n$ with confidence parameter $\delta$ that draws $s(\delta)$ examples and  runs in time $t(\delta)$.  Below is our algorithm for \lpn. Here, $H$ denotes the binary entropy function $p \mapsto p \log_2(1/p) + (1-p) \log_2(1/(1-p))$. 
\begin{algorithm}
\title{\textsc{Noisy}$(\delta, \eta)$}
\begin{algorithmic}[1]
\State Draw $s' = 20 s(\delta/2) \log(1/\delta)$ random examples $(a_1, \ell_1), \dots, (a_{s'}, \ell_{s'}) \in \zo^n \times \zo$. 
\ForAll{$S \subseteq [s'], |S| \leq \frac{3}{2}\eta s'$}
\For{$i \in [s']$}
\If{$i \in S$}
$\tilde{\ell}_i \gets 1-\ell_i$
\Else
$~\tilde{\ell}_i \gets \ell_i$
\EndIf
\EndFor
\State $x_S \gets \cA(\delta/2)$ applied to examples $(a_1, \tilde{\ell}_1), \dots, (a_{s'}, \tilde{\ell}_{s'})$. 
\EndFor
\State Draw $s'' = 600(s' \cdot H(3\eta/2) + \log(8/\delta))$ random examples $(b_1, m_1), \dots, (b_{s''}, m_{s''}) \in \zo^n \times \zo$
\State $S^* \gets \arg \max_{S \subset [s'], |S| \leq 3\eps s'/2} |\{i \in [s''] : \langle b_i, x_S\rangle = m_i\}|$
\State \Return $x_{S^*}$
\end{algorithmic}
\end{algorithm}

\begin{proof}[Proof of \cref{thm:red}]
\begin{lemma}
The sample complexity of \textsc{Noisy} is $s' + s'' = O(s(\delta/2)\log(1/\delta))$. 
\end{lemma}
\begin{proof}
Immediate.
\end{proof}

\begin{lemma}
The running time of \textsc{Noisy} is  $2^{O(H(3\eta/2)s(\delta/2)\log(1/\delta))}\cdot (t(\delta/2) + s(\delta/2) \log(1/\delta)))$.
\end{lemma}
\begin{proof}
We use the standard estimate $\sum_{i=0}^{\alpha x} {x \choose i} \leq 2^{H(\alpha)x}$ for $\alpha \leq \frac12$. The bound is then immediate. 
\end{proof}

\begin{lemma}
If $x$ is the hidden vector and $x^*$ is output by $\textsc{Noisy}(\delta)$, then with probability at least $1-\delta$, $x^* = x$. 
\end{lemma}
\begin{proof}
We make the assumption throughout that $\eta = \Omega(1/s(\delta/2))$, as otherwise, with high probability, the example source won't mislabel any of $s(\delta/2)$ samples and so the reduction is trivial.

Let $T = \{i \in [s'] : \langle a_i, x\rangle \neq \ell_i\}$ be the subset of the $s'$ samples drawn in line 1 that are mislabeled by the example source. By the Chernoff bound:
$\Pr[|T| > 3\eta s'/2] \leq e^{-\eta s'/12} \leq \delta/4$.
If $|T| \leq 3\eta s'/2$, we have with probability at least $1-\delta/2$, $x_T = x$. Thus, for any $i \in [s''], \Pr_{b_i}[\langle x_T, b_i\rangle \neq m_i] \leq \eta$. On the other hand, for all $x_S \neq x_T$, $\Pr_{b_i}[\langle x_S, b_i\rangle \neq \langle x, b_i \rangle]= 1/2$, and so $\Pr_{b_i}[\langle x_S, b_i \rangle \neq m_i ] =1/2$ as the noise is random. Again, using Chernoff bounds, 
$$Pr[\exists S \neq T \text{ s.t.} |\{i \in [s''] : \langle b_i, x_S\rangle \neq m_i\}| \leq 5s''/12] \leq 2^{H(3\eta/2)s'} \cdot e^{-s''/450} < \frac\delta 8$$
On the other hand, for $x_T$ itself,  $\Pr[|\{i \in [s''] : \langle b_i, x_T\rangle \neq m_i\}| > 5s''/12] < \frac\delta 8$ by a similar use of Chernoff bounds. So, in all, with probability at least $1-\delta$, $x_T$ will be returned in step 12. 
\end{proof}
\end{proof}

When the noise rate $\eta$ is more than $1/3$, a similar reduction can be given by adjusting the parameters accordingly. Also, when the distribution is arbitrary but the noise rate is less than $1/4$, a similar reduction can be made to work. In the latter case, $\cA$ is invoked with a smaller approximation parameter than the one given to \textsc{Noisy} so that the filtering step in line 10 works.

\subsection{Applications}
An immediate application of \cref{thm:red} is obtained by letting $\cA$ be the current fastest known attribute-efficient algorithm for learning \park, the algorithm due to Spielman\footnote{Though a similar algorithm was also proposed by Hopper and Blum \cite{HB01}} \cite{KS06} that makes $O(k \log n)$ samples and takes $O({n \choose k/2})$ time (for constant confidence parameter $\delta$). (We ignore the confidence parameter in this section for simplicity.) \\

{\noindent
\textbf{Corollary \ref{cor:att} (recalled)}
{\em For any $\eta \in (0,\nfrac13)$ and constant confidence parameter, there is an algorithm for \lpn with sample complexity $O(k \log n)$ and running time ${n \choose k/2}^{1 + O(H(1.5\eta))}$.} 
}
\begin{proof}
Immediate from \cref{thm:red}.
\end{proof}

Our next application of \cref{thm:pac} uses our improved \park learning algorithm from \cref{sec:nless}.\\

{\noindent
\textbf{Corollary \ref{cor:mainapp} (recalled)}
{\em Suppose $k(n) = n/f(n)$ for some function $f: \N \to \N$ for which $1 \ll f(n) \ll n/\log\log n$, and suppose $\eta(n) = o(1/((f(n))^\alpha \log n))$ for some $\alpha\in [1/2,1)$. Then, for constant confidence parameter, there exists an algorithm for \lpn with noise rate $\eta$ with running time $e^{-k/4.01 + o(k)} \cdot {n \choose k}^{1-\alpha} \cdot \poly(n)$ and sample complexity $O(k (f(n))^\alpha)$. }
}
\begin{proof}
Let $\cA$ be the algorithm of \cref{thm:pac} with $t(n) = \lceil n/(f(n))^\alpha \rceil$. The running time of $\cA$ is $e^{-k/4.01} \cdot {n \choose k}^{1-\alpha} \cdot \poly(n)$ and its sample complexity is $O(k\cdot f(n))^\alpha)$. Now, applying \cref{thm:red}, we see that since $H(1.5\eta) = o((f(n))^{-\alpha})$, the running time for \textsc{Noisy} is only a $2^{o(k)}$ factor times the running time of $\cA$. This yields our desired result.

\end{proof}

\bibliographystyle{alpha}
\bibliography{papers}

\end{document}

%% file: stdinc.tex
\usepackage{xspace}
\usepackage{verbatim}
\usepackage{color}
\usepackage{graphicx}
\usepackage{tabularx}

\usepackage{amsmath}
\usepackage[showonlyrefs]{mathtools}
\usepackage{mathstyle}
\usepackage{empheq}
\usepackage{amstext,amssymb,amsfonts}
\usepackage{fullpage}
\usepackage{nicefrac}
\usepackage{bm}

\usepackage{todonotes}

\newcounter{mynotes}
\usepackage{textcomp,setspace}

\usepackage{nameref}
\usepackage[pagebackref,colorlinks,linkcolor=blue,filecolor = blue,citecolor = blue, urlcolor = blue, hyperfootnotes=false]{hyperref}
\usepackage{color}
\usepackage[capitalize]{cleveref}
\renewcommand{\cref}{\Cref}

\usepackage{amsthm}
\usepackage{thmtools}

\declaretheorem[within=section]{theorem}
\declaretheorem[sibling=theorem]{corollary}

\declaretheorem[sibling=theorem]{lemma}
\declaretheorem[sibling=theorem]{claim}

\newcounter{termcounter}
\renewcommand{\thetermcounter}{\Alph{termcounter}}
\crefname{term}{term}{terms}
\creflabelformat{term}{\textup{(#2#1#3)}}

\makeatletter
\def\term{\@ifnextchar[\term@optarg\term@noarg}
\def\term@optarg[#1]#2{%
  \textup{(#1)}%
  \def\@currentlabel{#1}%
  \def\cref@currentlabel{[][2147483647][]#1}%
  \cref@label[term]{#2}}
\def\term@noarg#1{%
  \refstepcounter{termcounter}%
  \textup{(\thetermcounter)}%
  \cref@label[term]{#1}}
\makeatother

\newcommand{\nfrac}{\nicefrac}

\newcommand{\ignore}[1]{}

\newcommand{\bits}{\{0,1\}}
\newcommand{\zo}{\bits}

\newcommand{\poly}{\mathrm{poly}}

\definecolor{DSred}{rgb}{1,0,0}

\renewcommand{\leq}{\leqslant}
\renewcommand{\geq}{\geqslant}

\renewcommand{\epsilon}{\varepsilon}
\newcommand{\eps}{\epsilon}

\newcommand{\N}{\mathbb{N}}

\newcommand{\cA}{\mathcal A}

\newcommand{\cC}{\mathcal C}
\newcommand{\cD}{\mathcal D}

\newcommand{\cS}{\mathcal S}

\newcommand{\Psymb}{{\bf Pr}}

\DeclareMathOperator*{\ProbOp}{\Psymb}

\renewcommand{\Pr}{\ProbOp}

